%% file: sn-article.tex
\documentclass[pdflatex,sn-mathphys,Numbered]{sn-jnl}

\usepackage{graphicx}%
\usepackage{multirow}%
\usepackage{amsmath,amssymb,amsfonts}%
\usepackage{amsthm}%
\usepackage{mathrsfs}%
\usepackage[title]{appendix}%
\usepackage{xcolor}%
\usepackage{textcomp}%
\usepackage{manyfoot}%
\usepackage{booktabs}%
\usepackage{algorithm}%
\usepackage{algorithmicx}%
\usepackage{algpseudocode}%
\usepackage{listings}%

\usepackage{subcaption}

\theoremstyle{thmstyleone}%
\newtheorem{theorem}{Theorem}
\theoremstyle{thmstyletwo}%

\theoremstyle{thmstylethree}%

\begin{document}

\title[Fast Fr\'echet Distance]{Walking Your Frog Fast in 4 LoC}


%
%
%

\author*[1]{\fnm{Nis} \sur{Meinert}}\email{Nis.Meinert@dlr.de}
\affil*[1]{\orgdiv{Institute of Communications and Navigation, Nautical Systems}, \orgname{German Aerospace Center (DLR)}, \orgaddress{\street{Kalkhorstweg 53}, \city{Neustrelitz}, \postcode{17235}, \country{Germany}}}


\abstract{Given two polygonal curves, there are many ways to define a notion of similarity between them. One popular measure is the Fr\'echet distance which has many desirable properties but is notoriously expensive to calculate, especially for non-trivial metrics. In 1994, Eiter and Mannila introduced the discrete Fr\'echet distance which is much easier to implement and approximates the continuous Fr\'echet distance with a quadratic runtime overhead. However, this algorithm relies on recursions and is not well suited for modern hardware. To that end, we introduce the Fast Fr\'echet Distance algorithm, a recursion-free algorithm that calculates the discrete Fr\'echet distance with a linear memory overhead and that can utilize modern hardware more effectively. We showcase an implementation with only four lines of code and present benchmarks of our algorithm running fast on modern CPUs and GPGPUs.}



\maketitle

\input{main}


\bibliography{sn-bibliography}

\newpage

\begin{appendices}
    \input{supplementaries}
\end{appendices}

\end{document}

%% file: main.tex
\section{Introduction}
The Fr\'echet distance was introduced by Maurice Fr\'echet in 1906 \citep{frechet06} and is well known as a fundamental metric in abstract spaces in the field of mathematics and computational geometry.
Over the years, it has found applications in various domains, e.g., as a distance measure between probability distributions it is equivalent to the Wasserstein-2 distance \citep{kantorovich60,wasserstein69} for the $\ell^2$-norm.
Often, it can be evaluated efficiently, e.g., \citet{dowson82} demonstrated that for normal distributions the Fr\'echet distance can be calculated explicitly and since then it is used a basis of various evaluation metrics such as the Fr\'echet inception distance \citep{heusel17}, Fr\'echet audio distance \citep{kilgour19}, or the Fr\'echet ChemNet distance \citep{preuer18}.

In the early 1990s, Alt and Godau were the first to apply the Fr\'echet distance in measuring the similarity of polygonal curves \citep{alt92,alt95}.
Here, however, there is no closed-form expression that would yield the Fr\'echet distance for this setting.
Instead, the problem can be solved with dynamic programming and finding efficient algorithms remains a vibrant area of research \citep{aronov06,avraham14,agarwal14,bringmann19}.
Our contribution to this ongoing exploration is an efficient algorithm for the discrete Fr\'echet distance that is optimized to run fast on modern hardware.
In particular, we reformulate the recursive algorithm proposed by \citet{eiter94} as an iterative, branchless algorithm and reduce its quadratic memory requirement to a linear one.
Our formulation requires just four lines of pseudo-code, excluding function signatures, and can be vectorized to take full advantage of modern CPU and GPGPU architectures.
To the best of our knowledge, we are the first to comprehensively propose these modifications although it is possible that others---while only referring to \citet{eiter94}---are already using some of the modifications in practice without explicitly noting.

The remainder of this paper is structured as follows:
First, we describe the continuous and discrete Fr\'echet distance for polygonal curves in Sec.~\ref{sec:frechet}.
In Sec.~\ref{sec:fast} we present our improved algorithm.
Subsequently, we outline how our solution can be parallelized in Sec.~\ref{sec:parallel} and summarize our contribution in Sec.~\ref{sec:summary}.

\section{The Fr\'echet Distance}
\label{sec:frechet}
Formally, let $p$ and $q$ be two curves in a metric space $\mathcal{S}$.
Then, the (continuous) Fr\'echet distance $\delta_\mathrm{F}$ between $p$ and $q$ is defined as the infimum over all possible continuous, non-decreasing, surjections $\alpha: [0, 1] \to [0, 1]$ and $\beta: [0, 1] \to [0, 1]$ of the maximum over all $t \in [0, 1]$ of the distance in $\mathcal{S}$ between $p(\alpha(t))$ and $q(\beta(t))$,
\begin{equation*}
    \delta_\mathrm{F}(p, q) = \inf_{\alpha,\beta} \; \max_{t \in [0, 1]} \left\| p(\alpha(t)) - q(\beta(t)) \right\|_\mathcal{S},
\end{equation*}
where $\| \cdot \|_\mathcal{S}$ is the distance function of $\mathcal{S}$.
Informally, this measure is often described as the minimal leash length, measured by $\| \cdot \|_\mathcal{S}$, of a man walking his dog if both are moving on trajectories $p(\alpha(t))$ and $q(\beta(t))$, respectively, where $t$ is interpreted as a measure of time.
In this interpretation the restriction that $\alpha(t)$ and $\beta(t)$ be non-decreasing means that neither the dog nor its owner can backtrack.
Adding this constraint makes the Fr\'echet distance unconditionally stronger than the Hausdorff distance in the sense that it is always greater or equal than their corresponding Hausdorff distance between two curves.

Although non-straightforward to compute, the Fr\'echet distance has been used successfully in various fields such as curve simplification \citep{agarwal05,bereg08,kreveld18,kerkhof19}, map-matching \citep{alt03a,driemel12,sharma17} and clustering \citep{buchin08b,besse16,buchin19b,buchin19c}.
The Fr\'echet distance also has applications in matching biological sequences \citep{wylie12}, analysing tracking data \citep{brakatsoulas05,buchin08a}, and matching coastline data \citep{mascret06}.
In particular for clustering, a major advantage of the Fr\'echet distance is that it fulfills the triangle inequality and therefore is not only a distance measure but also a metric. 
Other measures of curve similarity, such as DTW~\citep{vintsyuk68,sankoff83} or the average Fr\'echet distance~\citep{brakatsoulas05}, do not obey the triangle inequality which significantly limits their application for downstream tasks in practice.

Computing the Fr\'echet distance between two precise curves can be done in near-quadratic time \citep{agarwal14,alt95,buchin17a}, and assuming the strong exponential time hypothesis, it cannot be computed or even approximated well in strongly subquadratic time unless SETH fails \citep{bringmann14,buchin19a}.
In practice, algorithms with a theoretical quadratic complexity often rely on certain pre-processing steps and are complex to implement in general.
Others, solve related problems such as the \textit{decider}, where the algorithm only determines if the Fr\'echet distance between curves is below a certain threshold, which does not help in downstream tasks where the exact distance is required, e.g., clustering \citep{bringmann19}.
Furthermore---and to the best of our knowledge---almost all of the analyses and applications of the continuous Fr\'echet distance between polygonal curves either explicitly or implicitly rely on the Euclidean distance for $\| \cdot \|_\mathcal{S}$ which makes the corresponding pathing in \textit{free-space} a convex optimization problem \citep{alt92,alt95}.
Extending the results to other metrics, e.g., the great-circle distance, is not straightforward as one quickly looses auxiliary properties such as the convexity.

One way to overcome these challenges is to estimate the continuous Fr\'echet distance of polygonal curves $\delta_\mathrm{F}$ with a discrete approximation $\delta_\mathrm{dF}$:
Instead of taking all points of the (continuous) curves into account, the discrete Fr\'echet distance only includes a finite number of them.
A typical choice for the set of the points is the set of vertices of the given polygonal curve itself, however, in order to decrease the approximation error, additional points can be sampled between them.
Informally and in comparison to the common \textit{man-walking-his-dog} analogy, the discrete Fr\'echet distance can be thought of as the minimal leash length between two frogs jumping between stones, where \textit{stones} are the vertices and the leash is not taken into consideration during jumps.

In 1994, \citet{eiter94} defined the discrete Fr\'echet distance $\delta_\mathrm{dF}$ between two $D$-dimensional polygonal curves $p \in \mathbb{R}^{P \times D}$ and $q \in \mathbb{R}^{Q \times D}$, proposed the simple Alg.~\ref{alg:vanilla} of complexity $\mathcal{O}(D'PQ)$, where $D'$ is the complexity of evaluating $\| \cdot \|_\mathcal{S}$, and showed that the difference between the continuous Fr\'echet distance and its discrete variant is bound by the sample width of the curves,
\begin{equation*}
    \delta_\mathrm{F}(p, q) \;\le\; \delta_\mathrm{dF}(p, q) \;\le\; \delta_\mathrm{F}(p, q) + \max\{ \varepsilon_p, \varepsilon_q \},
\end{equation*}
where $\varepsilon_p$ and $\varepsilon_q$ are the largest distance between adjacent points (the \textit{stones}) on $p$ and $q$, respectively.

\begin{algorithm}
    \caption{The original algorithm for calculating the discrete Fr\'echet distance as proposed by \citet{eiter94} for a generic distance measure $f: \mathbb{R}^D \times \mathbb{R}^D \mapsto \mathbb{R}$.}
    \label{alg:vanilla}
    \begin{algorithmic}[1]
        \Function{fr\'echet\_distance}{$p$: $\mathbb{R}^{P \times D}$, $q$: $\mathbb{R}^{Q \times D}$, $f$: $\mathbb{R}^D \times \mathbb{R}^D \mapsto \mathbb{R}$} $\to \mathbb{R}$
            \State $M: \mathbb{R}^{P \times Q}$ \label{alg:vanilla:M}
            \Statex
            \Function{eval}{$i: \mathbb{N}$, $j: \mathbb{N}$} $\to \mathbb{R}$
                \If{$M_{ij} > -1$}
                    \State \Return $M_{ij}$
                \EndIf
                \Statex
                \State $d: \mathbb{R}$ $\gets$ \Call{$f$}{$p_i, q_j$}
                \If{$i = 1 \;\land\; j = 1$} \label{alg:vanilla:if1}
                    \State $M_{ij} \gets d$
                \ElsIf{$i > 1 \;\land\; j = 1$} \label{alg:vanilla:if2}
                    \State $M_{ij} \gets \max\{\Call{eval}{i - 1, 1},\, d\}$
                \ElsIf{$i = 1 \;\land\; j > 1$} \label{alg:vanilla:if3}
                    \State $M_{ij} \gets \max\{\Call{eval}{1, j - 1},\, d\}$
                \Else \label{alg:vanilla:if4}
                    \State $M_{ij} \gets \max\{\min\{\Call{eval}{i - 1, j},\, \Call{eval}{i - 1, j - 1},\, \Call{eval}{i, j - 1}\},\, d\}$
                \EndIf
                \Statex
                \State \Return $M_{ij}$
            \EndFunction
            \Statex
            \State \Return \Call{eval}{$P, Q$}
        \EndFunction
    \end{algorithmic}
\end{algorithm}

Since then, the discrete Fr\'echet distance is used frequently, e.g., \citet{sriraghavendra07} have used it for handwriting recognition and \citet{jiang08} used the discrete Fr\'echet distance to tackle the protein structure-structure alignment problem.
Arguably, for the latter utilizing the discrete Fr\'echet distance even makes more sense than the continuous Fr\'echet distance as the backbone of a protein is simply a polygonal chain in 3D, with each vertex being the alpha-carbon atom of a residue.
Hence, if the continuous Fr\'echet distance is realized by an alpha-carbon atom and some other point which does not represent an atom, it is not meaningful biologically.
Later, \citet{zhu07} extended this work and analyzed the protein local structural alignment problem using bounded discrete Fr\'echet distance.

Furthermore, the Fr\'echet distance serves as a pivotal metric for assessing similarities among vessel trajectories in maritime research.
Here, the discrete Fr\'echet distance again emerges as a more fitting measure compared to its continuous variant, aligning seamlessly with the nature of trajectory data derived from discrete GNSS updates.
If not used for benchmarking alternative approaches \citep{li20,liu23,luo23}, in this context the discrete Fr\'echet distance is used a fundamental building block for movement analytics, clustering, and classification of the trajectories, for example:
\citet{sakan23} studied route characteristics by using the discrete Fr\'{e}chet distance to measure the similarity of individual container fleets from a statistical perspective.
\citet{cao18} used the discrete Fr\'{e}chet distance to solve the maximum distance between vessel trajectories and obtained a distance matrix, which was then decomposed using PCA to determine trajectory cluster.
Similarly, \citet{roberts18} constructed a graph where vessel trajectories are the vertices and the edges are weighted by the discrete Fr\'echet distance between two vertices.

Although applied successfully in practice, the algorithm by \citet{eiter94} has three major disadvantages:
First, the recursive formulation makes the algorithm impractical for larger curves as the recursive calls cannot be elided by tail recursions, thus making the algorithm slow and prone to stack overflows.
Secondly, the algorithm needs to allocate a memory block $M$ of size $P \times Q$ (see line~\ref{alg:vanilla:M} in Alg.~\ref{alg:vanilla}).
For large curves this allocation is slow, pollutes cache lines, and potentially constrains parallelization due to limited memory resources.
Thirdly, the explicit branching points on lines~\ref{alg:vanilla:if1}, \ref{alg:vanilla:if2}, \ref{alg:vanilla:if3} and \ref{alg:vanilla:if4} in Alg.~\ref{alg:vanilla} put stress on the branch predictor and make it hard to reformulate the algorithm for an effective use on \textit{single instruction, multiple data} (SIMD) or \textit{single instruction, multiple threads} (SIMT) architectures \citep{flynn72} without causing significant branch divergences.

\section{An Iterative, Linear-Memory Algorithm}
\label{sec:fast}
In this section we introduce a new algorithm for finding the discrete Fr\'echet distance between two polygonal curves $p \in \mathbb{R}^{P \times D}$ and $q \in \mathbb{R}^{Q \times D}$.
Our algorithm, Alg.~\ref{alg:fast}, uses (left) fold/reduce operators for the dyadic functions \textsc{fr\'echet\_maxmin}$(\cdot, \cdot)$ and $\max\{ \cdot, \cdot \}$ on lines \ref{alg:fast:scan1} and \ref{alg:fast:scan2}, and a (left) scan/accumulate operator for the dyadic function \textsc{fr\'echet\_next}$(\cdot, \cdot)$ on line~\ref{alg:fast:fold}, where we borrowed the notation for the operators $/$ and $\backslash$ from \citet{iverson62,iverson79}.
In particular, note that our scan operator prepends the initial value $\max\{ a_1, x_1 \}$ to the returned sequence.
On line \ref{alg:fast:scan2}, the parameters of the dyadic function $\max\{ \cdot, \cdot \}$ have the same type and thus we implicitly take the first value of the passed sequence as the initial value and pass the remaining sequence as the second argument.
Pseudocodes for the operators scan and fold are given in the Appendix in Alg.~\ref{alg:scan} and Alg.~\ref{alg:fold}, respectively.

\begin{theorem}[Fast Discrete Fr\'echet Distance Algorithm]
\label{thm:fast}
Alg.~\ref{alg:fast} calculates the discrete Fr\'echet distance given a distance matrix $d \in \mathbb{R}^{P \times Q}$.
The algorithm iteratively consumes the rows of $d$ such that each row, or even each element, can be computed lazily; the memory requirement is therefore reduced to $\mathcal{O}(Q)$.
The complexity of Alg.~\ref{alg:fast} is $\mathcal{O}(PQ)$ if $d$ is precomputed and $\mathcal{O}(D'PQ)$ if $d$ is evaluated lazily, where $D'$ is the complexity of evaluating a single element of $d$.

\begin{algorithm}
    \caption{A fast and concise algorithm that uses the fold of \textsc{fr\'echet\_next}$(\cdot, \cdot)$ (line~\ref{alg:fast:fold}) and the scans of \textsc{fr\'echet\_maxmin}$(\cdot, \cdot)$ (line~\ref{alg:fast:scan1}) and $\max \{ \cdot, \cdot \}$ (line~\ref{alg:fast:scan2}) to map a given distance matrix $d_{ij} = \| p_i - q_j \|_\mathcal{S}$ to the discrete Fr\'echet distance of $p \in \mathbb{R}^{P \times D}$ and $q \in \mathbb{R}^{Q \times D}$. The rows of $d$ can be evaluated lazily to achieve the linear memory requirement. The minimum on line~\ref{alg:fast:min} is taken element-wise. On line~\ref{alg:fast:scan1}, the column vectors $a_{2 \ldots Q}$ and $x_{2 \ldots Q}$ are stacked into a $\mathbb{R}^{(Q - 1) \times 2}$ matrix and iterated row-wise.}
    \label{alg:fast}
    \begin{algorithmic}[1]
        \Function{fr\'echet\_maxmin}{$a: \mathbb{R}$, $x: \mathbb{R}^2$} $\to \mathbb{R}$
            \State \Return $\max\{ \min\{ a,\, x_1 \},\, x_2 \}$
        \EndFunction
        \Statex
        \Function{fr\'echet\_next}{$a: \mathbb{R}^Q$, $x: \mathbb{R}^{Q}$} $\to \mathbb{R}^Q$
            \State $a_{2 \ldots Q} \gets \min\{ a_{1 \ldots Q-1},\, a_{2 \ldots Q} \}$ \label{alg:fast:min}
            \State \Return \Call{fr\'echet\_maxmin$\backslash$}{$\max\{a_1, x_1\},\, [a_{2 \ldots Q} \;|\; x_{2 \ldots Q}]$} \label{alg:fast:scan1}
        \EndFunction
        \Statex
        \Function{fr\'echet\_distance}{$d: \mathbb{R}^{P \times Q}$} $\to \mathbb{R}$
            \State \Return $[$\Call{fr\'echet\_next$/$}{$\max\!\backslash(d_1), d_{2 \ldots P}$}$]_Q$ \label{alg:fast:fold} \label{alg:fast:scan2}
        \EndFunction
    \end{algorithmic}
\end{algorithm}
\end{theorem}

\begin{proof}[Proof of Theorem~\ref{thm:fast}]
In order to prove Theorem~\ref{thm:fast} we start by rewriting Alg.~\ref{alg:vanilla} and replace the recursive calls with two explicit iterations over the points in $p$ and $q$.
Alg.~\ref{alg:vanilla} finds the discrete Fr\'echet distance between $p$ and $q$ \textit{top-down} by recursively solving the dynamic programming problem
\begin{equation}
    \label{eq:dynprogprob}
    M_{ij} = \max\{ \min\{ M_{i-1,j}, M_{i-1,j-1}, M_{i,j-1} \}, d_{ij} \},
\end{equation}
where $\delta_\mathrm{dF} = M_{PQ}$.
Alg.~\ref{alg:no_recursion} shows an algorithm that also solves Eq.~\eqref{eq:dynprogprob} but with a \textit{bottom-up} approach without recursions:
The nested loops start at low indices and eventually accumulate the result in $M_{ij}$ until $i=P$ and $j=Q$.
In fact, $M_{\alpha \beta}$ is the discrete Fr\'echet distance for the polygonal curves $p_i$ and $q_j$ with $i \in [1, \alpha]$ and $j \in [1, \beta]$, respectively, at every point.
The complexity of this algorithm is $\mathcal{O}(D'PQ)$ if the complexity of evaluating $f(p_i, q_j)$ on line \ref{alg:no_recursion:d} is $\mathcal{O}(D')$.

We believe that at least this simplification is already well established since \citet{ahn10} used a decision algorithm that is a close adaptation of Alg.~\ref{alg:vanilla} and is already implemented without recursions---yet still with a quadratic memory requirement and branches.
In Sec.~\ref{sec:dtw_levenshtein} we will discuss the relation of the discrete Fr\'echet distance to DTW and the Levenshtein distance.
Here as well iterative algorithms without recursions are common.

\begin{algorithm}
    \caption{Variant of Alg.~\ref{alg:vanilla} without recursions.}
    \label{alg:no_recursion}
    \begin{algorithmic}[1]
        \Function{fr\'echet\_distance}{$p: \mathbb{R}^{P \times D}$, $q: \mathbb{R}^{Q \times D}$, $f: \mathbb{R}^D \times \mathbb{R}^D \mapsto \mathbb{R}$} $\to \mathbb{R}$
            \State $M: \mathbb{R}^{P \times Q}$
            \Statex
            \For{$i \gets 1,P$}
                \For{$j \gets 1,Q$}
                    \State $d: \mathbb{R} \gets f(p_i, q_j)$ \label{alg:no_recursion:d}
                    \Statex
                    \If{$i = 1 \;\land\; j = 1$}
                        \State $M_{11} \gets d$
                    \ElsIf{$i > 1 \;\land\; j = 1$}
                        \State $M_{i,1} \gets \max\{M_{i-1,1},\, d \}$
                    \ElsIf{$i = 1 \;\land\; j > 1$}
                        \State $M_{1,j} \gets \max\{M_{1,j-1},\, d \}$
                    \Else
                        \State $M_{ij} \gets \max\{\min\{M_{i-1,j},\, M_{i-1,j-1},\, M_{i,j-1}\},\, d \}$
                    \EndIf
                \EndFor
            \EndFor
            \Statex
            \State \Return $M_{PQ}$   
        \EndFunction
    \end{algorithmic}
\end{algorithm}

Next, we remove branching points by evaluating the elements of $M_{ij}$ where either $i=1$ or $j=1$ before entering the nested loops with scans/accumulates of the first column (line~\ref{alg:inplace:scan1} of Alg.~\ref{alg:inplace}) and first row (line~\ref{alg:inplace:scan2} of Alg.~\ref{alg:inplace}) of $d$ using $\max\{ \cdot, \cdot \}$.
We note that this variant can be formulated as an in-place algorithm that operates on the distance matrix $d \in \mathbb{R}^{P \times Q}$ of $p$ and $q$ without needing any further allocations.
We show this variant in Alg.~\ref{alg:inplace}.

\begin{algorithm}
    \caption{In-place variant of Alg.~\ref{alg:vanilla}. This variant directly maps a given distance matrix $d_{ij} = \|p_i - q_j\|_\mathcal{S}$  to the corresponding Fr\'echet distance. Note that this variant has no explicit branching points.}
    \label{alg:inplace}
    \begin{algorithmic}[1]
        \Function{fr\'echet\_distance}{$d: \mathbb{R}^{P \times Q}$} $\to \mathbb{R}$
            \State $d_{:,1} \gets \max\!\backslash(d_{:,1})$ \label{alg:inplace:scan1}
            \State $d_{1,:} \gets \max\!\backslash(d_{1,:})$ \label{alg:inplace:scan2}
            \Statex
            \For{$i \gets 2,P$}
                \For{$j \gets 2,Q$}
                    \State $d_{ij} \gets \max\{\min\{d_{i-1,j},\, d_{i-1,j-1},\, d_{i,j-1}\},\, d_{ij} \}$
                \EndFor
            \EndFor
            \Statex
            \State \Return $d_{PQ}$   
        \EndFunction
    \end{algorithmic}
\end{algorithm}

Finally, we note that only two adjacent rows of $M$ (or $d$ in Alg.~\ref{alg:inplace}) are needed during the iterations.
In Alg.~\ref{alg:linear} we extract these rows, merge them into a single array $v \in \mathbb{R}^Q$ and thus reduce the overall memory requirement to $\mathcal{O}(Q)$.
In Alg.~\ref{alg:fast} the iterations over $p_i$ and $q_j$ have been rewritten as a fold/reduce and a scan/accumulate operation, respectively, and we adopted the concise function signature of Alg.~\ref{alg:inplace} for the sake of brevity.
\end{proof}

\begin{algorithm}
    \caption{Variant of Alg.~\ref{alg:vanilla} without recursions and explicit branching points, and a linear memory requirement.}
    \label{alg:linear}
    \begin{algorithmic}[1]
        \Function{fr\'echet\_distance}{$p: \mathbb{R}^{P \times D}$, $q: \mathbb{R}^{Q \times D}$, $f: \mathbb{R}^D \times \mathbb{R}^D \mapsto \mathbb{R}$} $\to \mathbb{R}$
            \State $v: \mathbb{R}^Q \gets \max\!\backslash(f(p_1, q))$
            \Statex
            \For{$i \gets 2,P$}
                \State $v_{2 \ldots Q} \gets \min\{ v_{1 \ldots Q-1},\, v_{2 \ldots Q} \}$
                \Statex
                \State $v_1 \gets \max\{ v_1, f(p_i, q_1) \}$
                \For{$j \gets 2,Q$}
                    \State $v_j \gets \max\{ \min\{ v_{j-1},\, v_j \},\, f(p_i, q_j) \}$ 
                \EndFor
            \EndFor
            \Statex
            \State \Return $v_Q$
        \EndFunction
    \end{algorithmic}
\end{algorithm}

The quadratic complexity of the algorithms can be improved if further approximation are made.
For example, \citet{aronov06} presented an efficient approximation algorithm for computing the discrete Fr\'echet distance of two natural classes of curves: $\kappa$-bounded curves and backbone curves.
They also proposed a pseudo-output-sensitive algorithm for computing the discrete Fr\'echet distance exactly.
However, even though the discrete Fr\'echet distance can be calculated faster for several restricted versions \citep{driemel12,gudmundsson18,buchin10,maheshwari11} our proposed algorithm works for the general case.
In particular, it does not make any assumptions about the metric or trajectory properties, such as Sakoe-Chiba bands, and scaling to higher dimensions only depends on $D'$.

\section{Parallel Implementations}
\label{sec:parallel}
Parallelizing \texttt{fr\'echet\_maxmin}$\backslash$, similar to the approaches introduced by \citet{kogge73} or \citet{brent82} for prefix sums, is not straightforward because \texttt{fr\'echet\_maxmin} is not associative.
However, due to the reduced memory requirement and the lack of branching points, the Fr\'echet distances between a batch of polygonal curves $\boldsymbol{p} \in [\mathbb{R}^{P_1 \times D}, \ldots, \mathbb{R}^{P_B \times D}]$ and a single polygonal curve $q \in \mathbb{R}^{Q \times D}$ can be easily evaluated in parallel on SIMD or SIMT architectures when curves within the batch are extended by repeating points\footnote{Repeating points of a polygonal curve does not change its Fr\'echet distance to others.} until $P_i = \tilde P \; \forall\, 1 \le i \le B$ such that $\boldsymbol{p}$ can be eventually written as a $\mathbb{R}^{\tilde P \times B \times D}$ tensor.

In case the variance of curve lengths within a batch is not too large, the benefit of evaluating the distance in batches compensates the cost of artificially extending curves and results in an improvement of the overall runtime.
Otherwise, this effect can be improved by sorting the curves by their respective lengths before assigning them to batches.

We benchmark the effectiveness of such a parallelization scheme by implementing Alg.~\ref{alg:fast} in \texttt{C++} using vectorization via SIMD instructions (CPU) and a CUDA implementation (GPGPU), and compare the performance with implementations of Alg.~\ref{alg:no_recursion} and \ref{alg:linear}.
Note that we intentionally not pick Alg.~\ref{alg:vanilla} for reference as this variant crashes due to its recursion nature quickly for larger trajectory sizes.

We conduct two experiments:
First, we generate $N$ random trajectories with $D=2$ of length $P = P_i = Q = 2^{10}$ and vary $N$ between $2^5$ and $2^{14}$.
Secondly, we keep $N$ fixed to $2^{10}$ and vary $P = P_i = Q$ between $2^5$ and $2^{14}$ points.
Each trajectory is generated by accumulating random steps in both dimensions, where each step, $(\Delta x, \Delta y)^\top \in \mathbb{R}^2$, is drawn uniformly from $\{-1, 0, +1\}$ in both directions, i.e., $\Delta x \times \Delta y \sim \{ -1, 0, +1 \} \times \{ -1, 0, +1 \}$.
We choose the Euclidean distance, $d_{ij} = \| p_i - q_j \|^2_2$, for the distance measure and evaluate it with the $\texttt{hypot}$ function of the $\texttt{C}$ standard library to simulate a non-trivial workload.

In order to unveil the full potential of vectorization, all implementations use 32-bit floating point numbers to represent (intermediate) steps during calculation.
We use an Ubuntu 22.04 machine using GCC-10 for the CUDA variant and GCC-11 for the rest.
The SIMD variant uses the \texttt{C++} extensions for parallelism as implemented in GCC-11~\citep{N4808}.
Builds where optimized using the flags \texttt{-O3}, \texttt{-DNDEBUG}, \texttt{-march=native}, and \texttt{-ffast-math}.
The reported numbers where taken after a warm-up phase during which the algorithm under investigation was run with the same data to minimize unwanted affects such as library loading, etc., during the measurements.

The results of the experiment are shown in Fig.~\ref{fig:benchmark_laptop}.
In Fig.~\ref{fig:benchmark_desktop} in the Appendix, we show additional results for a different hardware setup.

\begin{figure}[htbp]
    \centering
    \begin{subfigure}{.49\textwidth}
        \includegraphics[width=.8\textwidth]{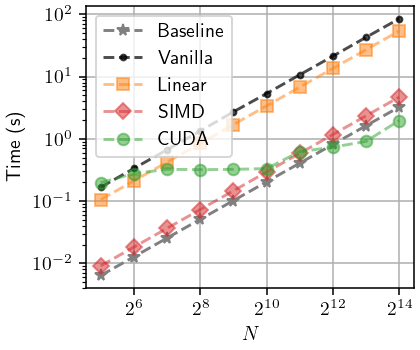}
        \caption{Absolute runtime with $P = 2^{10}$.}
    \end{subfigure}
    \begin{subfigure}{.49\textwidth}
        \includegraphics[width=.8\textwidth]{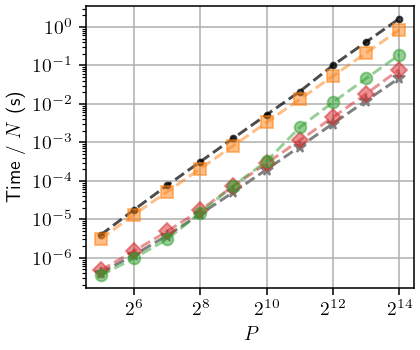}
        \caption{Absolute runtime with $N = 2^{10}$.}
    \end{subfigure}\\
    \begin{subfigure}{.49\textwidth}
        \includegraphics[width=.8\textwidth]{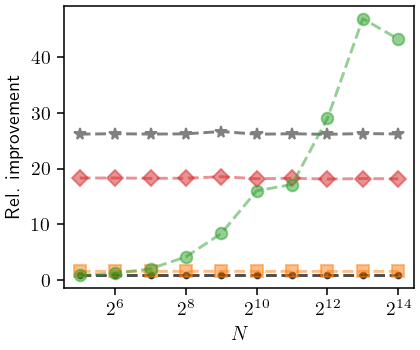}
        \caption{Relative improvement with $P = 2^{10}$.}
    \end{subfigure}
    \begin{subfigure}{.49\textwidth}
        \includegraphics[width=.8\textwidth]{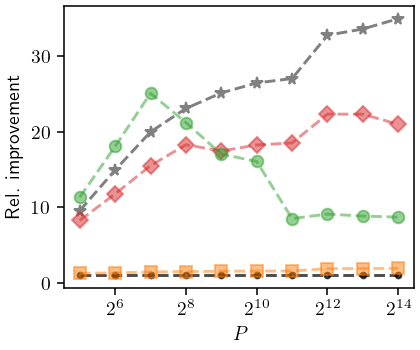}
        \caption{Relative improvement with $N = 2^{10}$.}
    \end{subfigure}
    \caption{Comparison of four different implementations of the Fast Fr\'echet Distance algorithm on a laptop (CPU: i7-11800H, GPU: NVIDIA GeForce RTX 3080 Mobile (16\,GB VRAM), CUDA Version: 12.2) using 32-bit floating point numbers. \textit{Vanilla} and \textit{Linear} refer to Alg.~\ref{alg:no_recursion} and \ref{alg:linear}, respectively. The SIMD implementation uses a batch size of $B = 32$ (twice the size of a 512-bit register in order to improve the instruction level parallelism) and relies on the AVX-512 instruction set; the baseline implementation uses the same technique to calculate $\sum_{ij} d_{ij}$. The CUDA implementation uses a grid and block size of 128 and 64, respectively, that was measured to perform best for $N=2^{13}$ and $P=2^{10}$. All variants utilize only a single CPU core.}
    \label{fig:benchmark_laptop}
\end{figure}

Besides the four variants, we further benchmark a baseline implementation that calculates the sum of all entries of the distance matrix $d \in \mathbb{R}^{P \times Q}$ using the very same vectorization techniques as the SIMD variant.
This gives us a decent approximation of an upper limit for the expected performance for the implementations running on the CPU, i.e., the difference between this baseline and the SIMD variant is the overhead induced by evaluating Eq.~\eqref{eq:dynprogprob} instead of simply summing the elements $d_{ij}$.

From Figs.~\ref{fig:benchmark_laptop} we see that the performance gain between Alg.~\ref{alg:vanilla} and \ref{alg:linear} is small but slightly increases when $P$ gets large.
In contrast, the SIMD variant comes close to our baseline and offers a constant improvement of roughly $19\times$ during the first and even up to $22\times$ during the second experiment.
We found that instruction level parallelism and an improved memory locality by packing $B=2 \times 16$ floating point numbers into 512-bit vector registers and using AVX-512 instructions explain this impressive boost.
Taking into account the large number of available threads of the GPU, the results of the CUDA implementation is less impressive and only outperforms the SIMD algorithm running on a single CPU core for large values of $N$.

We leave it to future works to experiment with more sophisticated scheduling approaches that replaces our na\"ive approach of artificially repeating points to match the $P_i = \tilde P \; \forall \, 1 \le i \le B$ requirement.
Furthermore, we conjecture that our algorithms can easily be extended to allow for MC sampling, e.g., to estimate the Fr\'echet distances for uncertain curves \citep{buchin23}.

\section{Summary}
\label{sec:summary}
In this paper, we have introduced the Fast Fr\'echet Distance algorithm in Alg.~\ref{alg:fast}, a recursion-free algorithm that calculates the discrete Fr\'echet distance with a linear memory overhead and can be implemented in four lines of code.
Our algorithm can easily be adopted, e.g., to estimate the distances between batches of polygonal curves and a reference curve in parallel.
We have shown benchmarks of implementations that efficiently utilize SIMD vector registers on a CPU and the parallelization potentials of a GPGPU.
Besides a possible application to uncertain curves, we further conjecture that even for sophisticated approaches that, for instance, use heuristics to avoid computing the full (discrete) Fr\'echet distance when not needed, will still benefit from our results as they are probably still bottlenecked by those cases where the heuristic fails \citep{buchin23,baldus17,buchin17b,duetsch17}.

\section{Reproducibility}
An open source GitHub repository with the source code for reproducing our experiments is available on \href{https://github.com/avitase/fast_frechet}{\texttt{github.com/avitase/fast\_frechet}}.
We encourage other researchers to reproduce, test, extend, and apply our work.

%% file: supplementaries.tex
\section{Scan and Fold}
\begin{algorithm}
    \caption{Scan operator for a binary function $f: \mathcal{T}_2 \times \mathcal{T}_1 \to \mathcal{T}_2$. The notation of using a backslash is borrowed from \citet{iverson62,iverson79}. Note that we prepend the initial value to the returned sequence. If $\mathcal{T}_1 = \mathcal{T}_2$, we take the first value of $x$ as the initial value.}
    \label{alg:scan}
    \begin{algorithmic}[1] \Function{$[f: \mathcal{T}_2 \times \mathcal{T}_1 \to \mathcal{T}_2]\backslash$}{$t_0: \mathcal{T}_2$, $x: \mathcal{T}_1^N$} $\to \mathcal{T}_2^{N+1}$
            \State $y: \mathcal{T}_2^{N+1}$
            \State $y_1 \gets t_0$
            \Statex
            \State $t: \mathcal{T}_2 \gets t_0$
            \For{$i \gets 1,N$}
                \State $t \gets f(t, x_i)$
                \State $y_{i+1} \gets t$
            \EndFor
            \Statex
            \State \Return $y$
        \EndFunction
        \Statex
        \Function{$[f: \mathcal{T} \times \mathcal{T} \to \mathcal{T}]\backslash$}{$x: \mathcal{T}^N$} $\to \mathcal{T}^N$
            \State \Return $f\backslash(x_1, x_{2 \ldots N})$
        \EndFunction
    \end{algorithmic}
\end{algorithm}

\begin{algorithm}
    \caption{Fold operator for a binary function $f: \mathcal{T}_2 \times \mathcal{T}_1 \to \mathcal{T}_2$. The notation of using a slash is borrowed from \citet{iverson62,iverson79}. If $\mathcal{T}_1 = \mathcal{T}_2$, we take the first value of $x$ as the initial value.}
    \label{alg:fold}
    \begin{algorithmic}[1]
        \Function{$[f: \mathcal{T}_2 \times \mathcal{T}_1 \to \mathcal{T}_2]/$}{$t_0: \mathcal{T}_2$, $x: \mathcal{T}_1^N$} $\to \mathcal{T}_2$
            \State $t: \mathcal{T}_2$ $\gets t_0$
            \For{$i \gets 1,N$}
                \State $t \gets f(t, x_i)$
            \EndFor
            \Statex
            \State \Return $t$
        \EndFunction
        \Statex
        \Function{$[f: \mathcal{T} \times \mathcal{T} \to \mathcal{T}]/$}{$x: \mathcal{T}^N$} $\to \mathcal{T}$
            \State \Return $f/(x_1, x_{2 \ldots N})$
        \EndFunction
    \end{algorithmic}
\end{algorithm}

\clearpage

\section{DTW and the Levenshtein Distance}
\label{sec:dtw_levenshtein}
Similarly to the Frechet distance, dynamic time warping (DTW) produces a discrete matching between existing elements of two polygonal curves.
The DTW distance is used frequently in the literature and fast algorithms are well-studied \citep{salvador07,praetzlich16,silva16,herrmann21}.
Replacing the dyadic function $\max\{ \cdot, \cdot \}$ with summation in Alg.~\ref{alg:fast} is sufficient to get an efficient algorithm for estimating the DTW distance and shows the close resemblance between both algorithms.
However, note that this change invalidates the triangle inequality of the resulting distant measure similar to the result of averaging the Fr\'echet distance \citep{brakatsoulas05}; hence, neither of these modifications result in metrics.
\begin{algorithm}
    \caption{Algorithm that maps a given distance matrix $d_{ij} = \| p_i - q_j \|_\mathcal{S}$ to the corresponding DTW distance of $p \in \mathbb{R}^{P \times D}$ and $q \in \mathbb{R}^{Q \times D}$.}
    \begin{algorithmic}[1]
        \Function{dtw\_min}{$a: \mathbb{R}$, $x: \mathbb{R}^2$} $\to \mathbb{R}$
            \State \Return $\min\{ a,\, x_1 \} + x_2$
        \EndFunction
        \Statex
        \Function{dtw\_next}{$a: \mathbb{R}^Q$, $x: \mathbb{R}^{Q}$} $\to \mathbb{R}^Q$
            \State $a_{2 \ldots Q} \gets \min\{ a_{1 \ldots Q-1},\, a_{2 \ldots Q} \}$
            \State \Return \Call{dtw\_min$\backslash$}{$a_1 + x_1,\, [a_{2 \ldots Q} \;|\; x_{2 \ldots Q}]$}
        \EndFunction
        \Statex
        \Function{dtw\_distance}{$d: \mathbb{R}^{P \times Q}$} $\to \mathbb{R}$
            \State \Return $[$\Call{dtw\_next$/$}{$+\backslash d_1, d_{2 \ldots P}$}
        \EndFunction
    \end{algorithmic}
\end{algorithm}

In Alg.~\ref{alg:levenshtein} we show another closely related distance measure: the Levenshtein distance~\cite{levenshtein65}---cf.~the results of \citet{wagner74} and \citet{hirschberg75}.
Although similar to the Fr\'echet or DTW distance, this algorithm does not search for the shortest, maximum distance as measured by $\| p_i - q_j \|_\mathcal{S}$, but rather accumulates $d_{ij} = [p_i \neq q_j] 
\in \{ 0, 1 \}$, where $[ \cdot ]$ is the Iverson bracket~\citep{knuth92} for character sequences $p \in \mathcal{C}^P$ and $q \in \mathcal{C}^Q$.

\begin{algorithm}
    \caption{Algorithm that maps a given distance matrix $d_{ij} = [p_i \neq q_j]$, where $[ \cdot ]$ is the Iverson bracket~\citep{knuth92}, to the Levenshtein distance of the character sequences $p \in \mathcal{C}^P$ and $q \in \mathcal{C}^Q$.}
    \label{alg:levenshtein}
    \begin{algorithmic}[1]
        \Function{levenshtein\_min}{$a: \mathbb{N}_+$, $x: \mathbb{N}_+$} $\to \mathbb{N}_+$
            \State \Return $\min\{ a + 1, x \}$
        \EndFunction
        \Statex
        \Function{levenshtein\_next}{$a: \mathbb{N}_+^Q$, $(i: \mathbb{N}_+,\, x: \mathbb{N}_+^Q)$} $\to \mathbb{N}_+^Q$
            \State $a_{2 \ldots Q} \gets \min\{a_{1 \ldots Q-1} + x_{2 \ldots Q},\, a_{2 \ldots Q} + 1 \}$
            \State \Return \Call{levenshtein\_min$\backslash$}{$\min\{ i + x_1,\, a_1 + 1 \},\, a_{2 \ldots Q}$}
        \EndFunction
        \Statex
        \Function{levenshtein\_distance}{$d: \mathbb{R}^{P \times Q}$} $\to \mathbb{N}_+$
            \State $\iota_{1 \ldots Q}: \mathbb{N}_+^{Q} \gets [1, \ldots, Q]$
            \State $\iota_{2 \ldots P}: \mathbb{N}_+^{P-1} \gets [2, \ldots, P]$
            \Statex
            \State $v_\mathrm{init}: \mathbb{N}_+ \gets$ \Call{levenshtein\_min$\backslash$}{$\iota_{1 \ldots Q} + d_1$}
            \State \Return \Call{levenshtein\_next$/$}{$v_\mathrm{init},\, (\iota_{2 \ldots P},\, d_{2 \ldots P})$}
        \EndFunction
    \end{algorithmic}
\end{algorithm}

\clearpage

\section{Additional Benchmarking Results}
\begin{figure}[htbp]
    \centering
    \begin{subfigure}{.49\textwidth}
        \includegraphics[width=.8\textwidth]{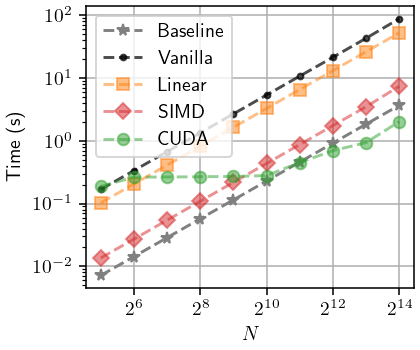}
        \caption{Absolute runtime with $P = 2^{10}$.}
    \end{subfigure}
    \begin{subfigure}{.49\textwidth}
        \includegraphics[width=.8\textwidth]{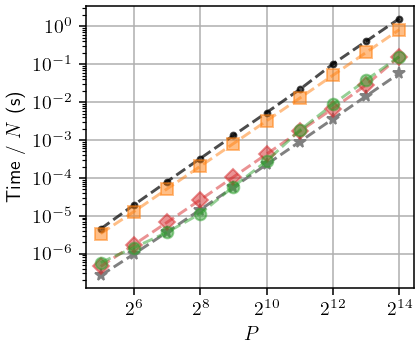}
        \caption{Absolute runtime with $N = 2^{10}$.}
    \end{subfigure}\\
    \begin{subfigure}{.49\textwidth}
        \includegraphics[width=.8\textwidth]{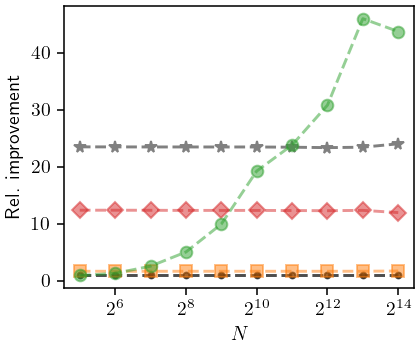}
        \caption{Relative improvement with $P = 2^{10}$.}
    \end{subfigure}
    \begin{subfigure}{.49\textwidth}
        \includegraphics[width=.8\textwidth]{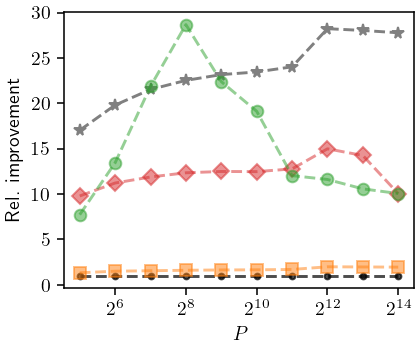}
        \caption{Relative improvement with $N = 2^{10}$.}
    \end{subfigure}
    \caption{Comparison of four different implementations of the Fast Fr\'echet Distance algorithm on a laptop (CPU: AMD Ryzen Threadripper 3960X, GPU: NVIDIA GeForce RTX 3090 (24\,GB VRAM), CUDA Version: 12.2) using 32-bit floating point numbers. \textit{Vanilla} and \textit{Linear} refer to Alg.~\ref{alg:no_recursion} and \ref{alg:linear}, respectively. The SIMD implementation uses a batch size of $B = 16$ (twice the size of a 256-bit register in order to improve the instruction level parallelism) and relies on the AVX2 instruction set; the baseline implementation uses the same technique to calculate $\sum_{ij} d_{ij}$. The CUDA implementation uses a grid and block size of 128 and 64, respectively, that was measured to perform best for $N=2^{13}$ and $P=2^{10}$. All variants utilize only a single CPU core.}
    \label{fig:benchmark_desktop}
\end{figure}